\documentclass[submission,copyright,creativecommons]{eptcs}
\providecommand{\event}{EPTCS} % Name of the event you are submitting to
\usepackage{breakurl}             % Not needed if you use pdflatex only.

\usepackage{proof}
\usepackage{amssymb}
\usepackage{amsmath}
\usepackage{amsthm}

\newtheorem{theorem}{Theorem} %[section]
\newtheorem{property}[theorem]{Property}
\newtheorem{proposition}[theorem]{Proposition}
\newtheorem{lemma}[theorem]{Lemma}
\newtheorem{corollary}[theorem]{Corollary}
\newtheorem{definition}[theorem]{Definition}

\newcommand{\barb}[1]{\mathop{{\Downarrow}_{#1}}}
\newcommand{\sbarb}[1]{\mathop{{\downarrow}_{#1}}}
\newcommand{\aRel}{\ensuremath{\mathrel{{\mathcal R }}}}
\newcommand{\ar}[1]{\mathrel{\stackrel{#1}{\longrightarrow}}}
\newcommand{\dar}[1]{\mathrel{\stackrel{#1}{\Longrightarrow}}}

\newcommand{\cxtleq}{\preceq_{c}}
\newcommand{\si}{\preceq_{s}}
\newcommand{\lleq}{\preceq_{l}}

\newcommand{\unit}{\textbf{1}}
\newcommand{\tensor}{\mathop{\otimes}}
\newcommand{\lolli}{\mathop{\multimap}}
\newcommand{\bang}{!}

\newcommand{\reduce}{\rightsquigarrow}
\DeclareMathOperator{\with}{\&}

\newcommand{\D}{\Delta}
\newcommand{\G}{\Gamma}
\newcommand{\lseq}[3][]{{#2} \vdash_{#1} {#3}}
\newcommand{\sep}{\;\mid\;}
\newcommand{\rname}[1]{\ensuremath{#1}}

\newcommand{\sset}[1]{\{ {#1}  \}  } % singleton set

\usepackage{color}

\title{%Process Interpretation of Linear Logic%
Relating Reasoning Methodologies in Linear Logic and Process Algebra%
\thanks{Partially supported by the Qatar National Research Fund under grant NPRP 09-1107-1-168, the Natural Science Foundation of China under grant 61173033, and the Funda\c{c}\~ao para a 
Ci\^encia e a Tecnologia (Portuguese Foundation for Science and Technology)
through the Carnegie Mellon Portugal Program under Grant NGN-44.
}}
\author{Yuxin Deng
        \institute{Carnegie Mellon University\\
                   and Shanghai Jiao Tong University}
        \email{yuxin@cs.cmu.edu}
\and    Iliano Cervesato
        \institute{Carnegie Mellon University\\
                   Doha, Qatar}
        \email{iliano@cmu.edu}
\and    Robert J. Simmons
        \institute{Carnegie Mellon University\\
                   Pittsburgh, PA}
        \email{rjsimmon@cs.cmu.edu}
}
\def\titlerunning{Relating Reasoning Methodologies in Linear Logic and Process Algebra}
\def\authorrunning{Y. Deng, I. Cervesato \& R.J. Simmons}
\begin{document}
\maketitle

\begin{abstract}
  We show that the proof-theoretic notion of logical preorder coincides with
  the process-theoretic notion of contextual preorder for a CCS-like calculus
  obtained from the formula-as-process interpretation of a fragment of linear
  logic.  The argument makes use of other standard notions in process algebra,
  namely a labeled transition system and a coinductively defined simulation
  relation.  This result establishes a connection between an approach to
  reason about process specifications and a method to reason about logic
  specifications.
\end{abstract}

\section{Introduction}
\label{sec:intro}

By now, execution-preserving relationships between (fragments of) linear logic
and (fragments of) process algebras are well-established
(see~\cite{cervesato09relating} for an overview).  Abramsky observed early on
that linear cut elimination resembles reduction in CCS and the
$\pi$-calculus~\cite{Mil89}, thereby identifying processes with (some) linear
proofs and establishing the \emph{process-as-term}
interpretation~\cite{abramsky94tcl}.  The alternative
\emph{process-as-formula} encoding, pioneered by Miller around the same
time~\cite{Miller92elp}, maps process constructors to logical connectives and
quantifiers, with the effect of relating reductions in process algebra with
proof steps, in the same way that logic programming achieves computation via
proof search.  This interpretation has been used extensively in a multitude of
domains~\cite{Cervesato00fmcs, cervesato02concurrent, cervesato09relating,
  Miller92elp}, e.g., programming languages and security.

Not as well established is the relationship between the rich set of notions
and techniques used to reason about process specifications and the equally
rich set of techniques used to reason about (linear) logic.  Indeed, a
majority of investigations have attempted to reduce some of the behavioral
notions that are commonplace in process algebra to derivability within logic.
For example, Miller identified a fragment of linear logic that could be used
to observe traces in his logical encoding of the $\pi$-calculus, thereby
obtaining a language that corresponds to the Hennessy-Milner modal logic, which
characterizes observational equivalence~\cite{Miller92elp}.  A similar
characterization was made in~\cite{lincoln91tr}, where a sequent $\Gamma
\vdash \Delta$ in a classical logic augmented with constraints was seen as
process state $\Gamma$ passing test $\Delta$.  Extensions of linear logic were
shown to better capture other behavioral relations: for example, adding
definitions allows expressing simulation as the derivability of a linear
implication~\cite{mcdowell03tcs}, but falls short of bisimulation, for which a
nominal logic is instead an adequate formalism~\cite{tiu04fguc}.

This body of work embeds approaches for reasoning \emph{about} process
specifications (e.g., bisimulation or various forms of testing) into methods
for reasoning \emph{with} logic (mainly derivability).  Little investigation
has targeted notions used to reason about logic (e.g., proof-theoretic
definitions of equivalence). 
% More generally, techniques and tools developed
% for each formalism rarely cross over --- and may very well be rediscovered in
% due time.  Tellingly, process-based definitions and proofs are often
% coinductive in nature, while techniques based on proof-theory are generally
% inductive.

This paper outlines one such relationship --- between the inductive methods
used to reason about logic and the coinductive methods used to reason about
process calculi.  On the linear logic side, we focus on the
inductively-defined notion of logical preorder; this preorder is novel in the
sense that it is a natural and proof-theoretic way of relating \emph{contexts}
to other contexts.  On the process-algebraic side, we consider an extensional
behavioral relation adapted from the standard coinductive notion of contextual
preorder.  We prove that, for a fragment of linear logic and a matching
process calculus, these notions coincide, and we hope in future work to extend
this result to a larger fragment of intuitionistic linear logic. Our proofs
rely on other standard process algebraic notions as stepping stones, namely
simulation and labeled transition systems.

The rest of the paper is organized as follows.  In Section~\ref{sec:ll}, we
briefly review the fragment of linear logic we are focusing on and define the
logical preorder.  Then, in Section~\ref{sec:ctx}, we recall its standard
process-as-formula interpretation and define the contextual preorder.  In
Section~\ref{sec:lts}, we prove their equivalence through the intermediary of
a simulation preorder defined on the basis of a labeled transition system.
Full proofs of all results in this paper can be found in the accompanying
technical report~\cite{deng11tr}.

\section{Logical Preorder}
\label{sec:ll}

\begin{figure}
\fbox{%
\parbox{0.98\linewidth}{%
\[
\infer[\it init]
{\mathstrut
 \lseq{\G;a}{a}
}
{}
\qquad
\infer[\it clone]
{\mathstrut
 \lseq{\G,A;\D}{C}
}
{\lseq{\G,A;\D,A}{C}}
\]

\[
\infer[{\tensor}R]
{\mathstrut
 \lseq{\G;\D_1, \D_2}{A \tensor B}
}
{\mathstrut
 \lseq{\G;\D_1}{A}
&
 \lseq{\G;\D_2}{B}
}
\qquad
\infer[{\tensor}L]
{\mathstrut
 \lseq{\G;\D, A \otimes B}{C}
}
{\mathstrut
 \lseq{\G;\D, A, B}{C}
}
\qquad
\infer[{\bf 1}R]
{\mathstrut
 \lseq{\G;\cdot}{{\bf 1}}
}
{}
\qquad
\infer[{\bf 1}L]
{\mathstrut
 \lseq{\G;\D, {\bf 1}}{C}
}
{\mathstrut
 \lseq{\G;\D}{C}
}
\]

\[
\infer[{\with}R]
{\mathstrut
 \lseq{\G;\D}{A \with B}
}
{\mathstrut
 \lseq{\G;\D}{A}
&
 \lseq{\G;\D}{B}
}
\qquad
\infer[{\with}L_i]
{\mathstrut
 \lseq{\G;\D, A_1 \with A_2}{C}
}
{\mathstrut
 \lseq{\G;\D, A_i}{C}
}
\qquad
\infer[{\top}R]
{\mathstrut
 \lseq{\G;\D}{\top}
}
{}
\qquad
\mbox{\it (no rule ${\top}L$)}
\]

\[
\infer[{\lolli}R]
{\mathstrut
 \lseq{\G;\D}{a \lolli B}
}
{\mathstrut
 \lseq{\G;\D, a}{B}
}
\qquad
\infer[{\lolli}L]
{\mathstrut
 \lseq{\G;\D_1, \D_2, a \lolli B}{C}
}
{\mathstrut
 \lseq{\G;\D_1}{a}
&
 \lseq{\G;\D_2, B}{C}
}
\qquad
\infer[!R]
{\mathstrut
 \lseq{\G;\cdot}{!A}
}
{\mathstrut
 \lseq{\G;\cdot}{A}
}
\qquad
\infer[!L]
{\mathstrut
 \lseq{\G;\D, !A}{C}
}
{\mathstrut
 \lseq{\G,A;\D}{C}
}
\]
}}
\caption{A Fragment of Dual Intuitionistic Linear Logic}
\label{fig:mall}
\end{figure}

The fragment of linear logic considered in this paper is given by the
following grammar:
$$
\begin{array}{lrl}
  A, B, C
 &  ::=
 &      a
   \sep \unit
   \sep A \tensor B
   \sep \top
   \sep A \with B
   \sep a \lolli B
   \sep !A
\end{array}
$$
where $a$ is an atomic formula.  This language is propositional and, as often
the case in investigations of CCS-like process algebras~\cite{Cervesato00fmcs,
  cervesato02concurrent, cervesato09relating}, the antecedent of linear
implication is restricted to atomic formulas (see the remarks in
Section~\ref{sec:concl} about lifting these constraints).

Derivability for this language is given in terms of dual intuitionistic linear
logic (DILL) sequents~\cite{barber96dual, cervesato09relating} of the form
$\lseq{\G;\D}{A}$, where the \emph{unrestricted context} $\G$ and the
\emph{linear context} $\D$ are multisets of formulas.  Formally, they are
defined by the productions %
$\G, \D ::= \cdot \sep \D, A$ %
where ``$\cdot$'' represents the empty context, and ``$\D,A$'' is the context
obtained by adding the formula $A$ to the context $\D$.  As usual, we
tacitly treat ``$,$'' as an associative and commutative context union operator
``$\D_1,\D_2$'' with ``$\cdot$'' as its unit.

The fairly standard inference rules defining derivability are given in
Figure~\ref{fig:mall}.  A DILL sequent $\lseq{\G;\D}{A}$ corresponds to
$\lseq{\bang\G,\D}{A}$ in Girard's original
presentation~\cite{girard87linear}.  We take the view, common in practice,
that the context part of the sequent $(\G;\D)$ represents the state of some
system component and that the consequent $A$ corresponds to some property
satisfied by this system.

% It is common practice to use the context part of a sequent as a logical
% specification of some system component, while the formula on the right-hand
% side is a property of this specification.  
We will be interested in a relation, the logical preorder, that compares
specifications on the basis of the properties they satisfy, possibly after the
components they describe are plugged into a larger system.  This relation,
written $\preceq_l$, is given by the following definition.

% We will be interested in a describing a relation, the logical preorder, that
% relates two system components $(\G_1;\D_1)$ and $(\G_2;\D_2)$ on the basis of
% the properties they satisfy, possibly after the components are plugged into a
% larger system. This relation, written $\preceq_l$, is given by the following
% definition.

\begin{definition}[Logical preorder]
\label{def:logical-preorder}
The \emph{logical preorder} is the smallest relation $\preceq_l$ 
such that $(\G_1;\D_1) \lleq(\G_2;\D_2)$ if, for all 
$\G'$, $\D'$, and $C$, we have that 
$\lseq{(\G',\G_1);(\D',\D_1)}{C}$ implies
$\lseq{(\G',\G_2);(\D',\D_2)}{C}$.
\qed
\end{definition}

This relation is reflexive and transitive, and therefore a
preorder~\cite[Theorem 2.5]{deng11tr}; 
we could define logical equivalence as the symmetric
closure of $\preceq_l$.  The above definition is extensional in the sense that
it refers to all contexts $\G'$ and $\D'$ and formulas $C$.  It has also an
inductive characterization based on derivability~\cite[Theorem 2.8]{deng11tr}:

\begin{property}
$(\G_1;\D_1) \lleq (\G_2;\D_2)$ \ iff \
$\lseq{\G_2;\D_2}{\bigotimes!\G_1 \,\tensor\: \bigotimes\D_1}$
\end{property}
Here, $\bigotimes\D_1$ denotes the conjunction of all formulas in $\D_1$ (or
$\unit$ if it is empty) and $!\G_1$ is the linear context obtained by
prefixing every formula in $\G_1$ with ``$!$''.

The logical preorder has other interesting properties, such 
as \emph{harmony}~\cite[Proposition 2.6]{deng11tr}:

\begin{property}[Harmony]
  $\lseq{\G;\D}{A}$ if and only if $(\cdot;A) \preceq_l (\G; \D)$.
\qed
\end{property}
In the accompanying technical report, we show that a deductive system
satisfies harmony if and only if the rules of \emph{identity} and \emph{cut}
are admissible:
$$
\infer[\it identity]
{\lseq{\G;A}{A}}
{}
\qquad
\infer[\it cut]
{\lseq{(\G,\G');(\D,\D')}{C}}
{\lseq{\G;\D}{A} & \lseq{\G';\D',A}{C}}
$$
In particular, this means that harmony holds not only for our restricted
fragment of DILL, but for full DILL and most other syntactic fragments of DILL
as well \cite[Section 2.3]{deng11tr}.

\section{Contextual Preorder}
\label{sec:ctx}

The subset of linear logic just introduced has a natural interpretation as a
fragment of CCS~\cite{Mil89} with CSP-style internal choice~\cite{Hoa85}. It
is shown in Figure~\ref{fig:paf}.  We will now switch to this reading, which
is known as the \emph{conjunctive process-as-formula interpretation} of linear
logic~\cite{Miller92elp, cervesato09relating}.  Therefore, for most of the
rest of this section, we understand $A$ as a process.

Under this reading, contexts $(\G;\D)$ are \emph{process states}, i.e.,
systems of parallel processes understood as the parallel composition of each
process in $\D$ and, after restoring the implicit replication, in $\G$.  In
process algebra, parallel composition is considered associative and
commutative and has the null process as its unit.  This endows process states
with the following structural congruences:
$$
\begin{array}{rcl@{\hspace{3em}}rcl}
   (\G;\:\D,\cdot)          & \equiv & (\G;\:\D)
 & (\G,\cdot;\: \D)         & \equiv & (\G;\:\D)
\\ (\G;\:\D_1,\D_2)         & \equiv & (\G;\:\D_2,\D_1)
 & (\G_1,\G_2;\: \D)        & \equiv & (\G_2,\G_1;\: \D)
\\ (\G;\:\D_1,(\D_2,\D_3))  & \equiv & (\G;\:(\D_1,\D_2),\D_3)
 & (\G_1,(\G_2,\G_3);\: \D) & \equiv & ((\G_1,\G_2),\G_3;\: \D)
\\&&
 & (\G,A,A;\: \D)           & \equiv & (\G,A;\: \D)
\end{array}
$$
which amount to asking that $\D$ and $\G$ be commutative monoids.  The last
equality on the right merges identical replicated processes, making $\G$ into
a set.  In the following, we will always consider process states modulo this
structural equality, and therefore treat equivalent states as syntactically
identical.

Opening a parenthesis back into logic, it easy to prove that structurally
equivalent context pairs are logically equivalent.  In symbols, if
$(\G_1;\D_1) \equiv (\G_2;\D_2)$, then $(\G_1;\D_1) \lleq (\G_2;\D_2)$ and
$(\G_2;\D_2) \lleq (\G_1;\D_1)$.

The analogy with CCS above motivates the reduction relation $\reduce$ between
process states defined in Figure~\ref{fig:mall2}.  A formula $A \tensor B$
(parallel composition) transitions to the two formulas $A$ and $B$ in
parallel, for instance, and a formula $A \with B$ (choice) either
transitions to $A$ or to $B$.  The rule corresponding to implication is also
worth noting: a formula $a \lolli B$ can interact with an atomic formula $a$
to produce the formula $B$; we think of the atomic formula $a$ as sending a
message asynchronously and $a \lolli B$ as receiving that message.  We write
$\reduce^*$ for the reflexive and transitive closure of $\reduce$.

\begin{figure}[t]
\begin{minipage}[t]{0.61\linewidth}
\fbox{%
\parbox[t]{0.90\linewidth}{%
\vspace*{-1.5ex}%
\[
\begin{array}{rp{18em}}
   a           & atomic process that sends $a$
\\ A \tensor B & process that forks into processes $A$ and $B$
\\ \unit       & null process
\\ A_1 \with A_2 & process that can behave either as $A_1$ or as $A_2$
%\\ 
\\ \top        & stuck process
\\ a \lolli B  & process that receives $a$ and continues as $B$
\\ \bang{A}    & any number of copies of process $A$
\end{array}\vspace*{-1.5ex}
\]
}}
\caption{Process-as-formula Interpretation}
\label{fig:paf}
\end{minipage}
\begin{minipage}[t]{0.38\linewidth}
\fbox{%
\parbox[t]{0.98\linewidth}{%
\vspace*{-4.2ex}%
\[
\begin{array}{rcl} %@{\hspace{2em}}l}
\\ (\G;\:\D, A \tensor B)   & \reduce & (\G;\:\D, A, B)  % & ({\reduce}{\tensor})
\\ (\G;\:\D, \unit)         & \reduce & (\G;\:\D)        % & ({\reduce}{\unit})
\\ (\G;\:\D, A_1 \with A_2) & \reduce & (\G;\:\D, A_i)   % & ({\reduce}{\with_i})
\\ \multicolumn{3}{c}{\text{\emph{(No rule for $\top$)}}}
\\ (\G;\:\D, a, a \lolli B) & \reduce & (\G;\:\D, B)     % & ({\reduce}{\lolli})
\\ (\G;\: \D, \bang{A})     & \reduce & (\G, A;\: \D)    % & (\reduce\bang)
\\ (\G,A;\: \D)             & \reduce & (\G,A;\: \D,A)   % & (\reduce\mathrm{clone})
\end{array}\vspace*{-1.5ex}
\]
}}
\caption{Transitions}
\label{fig:mall2}
\end{minipage}
\end{figure}

Intuitively, two systems of processes are contextually equivalent if they
behave in the same way when composed with any given process.  We will be
interested in the asymmetric variant of this notion, a relation known as the
contextual preorder.  We understand ``behavior'' as the ability to produce the
same messages.  To model this, we write $(\G;\D)\sbarb{a}$ whenever $a\in\D$
and $(\G;\D)\barb{a}$ whenever $(\G;\D)\reduce^*(\G';\D')$ for some $(\G';\D')$
with $(\G';\D')\sbarb{a}$.  We also define the \emph{composition} of two
states $(\G_1;\D_1)$ and $(\G_2;\D_2)$, written $((\G_1;\D_1), (\G_2;\D_2))$,
as the state $((\G_1, \G_2); (\D_1,\D_2))$.  
%Recall that persistent contexts are sets and may collapse identical formulas, while linear contexts can contain duplicates.

In defining our contextual preorder, we will also require that partitions of
systems of processes behave congruently --- a related notion, Markov
simulation, is used in probabilistic process algebras~\cite{DH11a}.  We
formally capture it by defining that a \emph{partition} of a state $(\G;\D)$
is any pair of states $(\G_1;\D_1)$ and $(\G_2;\D_2)$ such that $(\G;\D) \equiv
((\G_1;\D_1),(\G_2;\D_2))$.

\begin{definition}[Contextual preorder]
  Let $\aRel$ be a binary relation over states.  We say that $\aRel$ is
  \begin{description}
  \item[\emph{barb-preserving}]%
    if, whenever $(\G_1;\D_1) \aRel (\G_2;\D_2)$ and $(\G_1;\D_1)\sbarb{a}$,
    we have that $(\G_2;\D_2)\barb{a}$ for any $a$.

  \item[\emph{reduction-closed}]%
    if $(\G_1;\D_1) \aRel (\G_2;\D_2)$ and $(\G_1;\D_1) \reduce
    (\G_1';\D'_1)$ implies $(\G_2;\D_2) \reduce^* (\G_2';\D'_2)$ and
    \linebreak[4]
    $(\G_1';\D'_1) \aRel (\G_2';\D'_2)$  for some $(\G_2';\D'_2)$.

  \item[\emph{compositional}]%
    if $(\G_1;\D_1) \aRel (\G_2;\D_2)$ implies $((\G_1;\D_1), (\G;\D)) \aRel
    ((\G_2;\D_2), (\G;\D))$ for all $(\G;\D)$.

 \item[\emph{partition-preserving}]%  
    if $(\G_1;\D_1) \aRel (\G_2;\D_2)$ implies that
    \begin{enumerate}
    \item%
      if $\D_1 = \cdot$, then $(\G_2;\D_2)\reduce^* (\G_2';\cdot)$ and
      $(\G_1;\cdot) \aRel (\G'_2;\cdot)$,
    \item %
      for all $(\G_1';\D'_1)$ and $(\G_1'';\D''_1)$, if
      $(\G_1;\D_1)=((\G_1';\D'_1),(\G_1'';\D''_1))$ then
      there exists $(\G_2';\D'_2)$ and $(\G_2'';\D''_2)$ such that
      $(\G_2;\D_2)\reduce^* ((\G_2';\D'_2),(\G_2'';\D''_2))$ and furthermore
      $(\G_1';\D'_1) \aRel (\G_2';\D'_2)$  and $(\G_1'';\D''_1) \aRel
      (\G_2'';\D''_2)$,
    \end{enumerate}
  \end{description}
The \emph{contextual preorder}, denoted by $\cxtleq$, is the largest
relation over processes which is barb-preserving, reduction-closed,
compositional and partition-preserving.
\qed
\end{definition}

The contextual preorder is indeed reflexive and transitive, and therefore a
preorder~\cite[Theorem~3.7]{deng11tr}. 
%???
In contrast to the proof of analogous property of the logical
preorder, the proof of this result is coinductive.

Contextual equivalence, which is the symmetric closure of the contextual
preorder, has been widely studied in concurrency theory, though its appearance
in linear logic seems to be new.  It is also known as \emph{reduction barbed
  congruence} and used in a variety of process calculi~\cite{ht92, RathkeS08,
  FournetG05, DH11a}.
%  There are also minor
% variations on the formulations. See~\cite{FournetG05, pibook} for a discussion
% of the differences.

\section{Correspondence via Simulation}
\label{sec:lts}

In this section, we show that the logical and the contextual preorders are the
same relation.  A direct proof eluded us, as the inductive reasoning
techniques that underlie derivations (on which $\preceq_l$ is based) do not
play nicely with the intrinsically coinductive arguments that are natural for
$\cxtleq$.  Instead, our proof uses a second relation, the simulation
preorder, as a stepping stone.  This intermediary relation is also coinductive,
but it is relatively easy to show that it is equivalent to the logical
preorder.  

The definition of the simulation preorder relies on the labeled transition
system in Figure~\ref{fig:lts-!}.  It defines the transition judgment
\smash{$(\G_1;\D_1) \ar{\beta} (\G_2;\D_2)$} between states $(\G_1;\D_1)$ and
$(\G_2;\D_2)$.  Here, $\beta$ is a label.  We distinguish ``non-receive''
labels, denoted $\alpha$, as either the silent action $\tau$ or a label $!a$
for atomic formula $a$ --- it represents a send action of $a$.  Generic labels
$\beta$ extend them with receive actions $?a$.  We write \smash{$\dar{\tau}$}
for the reflexive and transitive closure of \smash{$\ar{\tau}$}, and
\smash{$(\G_1;\D_1) \dar{\beta} (\G_2;\D_2)$} for \smash{$(\G_1;\D_1)
  \dar{\tau}\ar{\beta}\dar{\tau} (\G_2;\D_2)$}, if $\beta\not=\tau$.

Rule \rname{lts!?} synchronizes a send action (rule \rname{lts!}) with a
receive action (rule \rname{lts?}), thereby achieving the same effect as the
transition for $\lolli$ in Figure~\ref{fig:mall2}.  The other $\tau$
transitions in Figure~\ref{fig:lts-!}  correspond directly to reductions
(since $\top$ is the stuck process, it has no action to perform).  Indeed, the
following result holds~\cite[Lemmas 4.1 and~5.2]{deng11tr}:

\begin{property}
 $(\G_1;\D_1) \dar{\tau} (\G_2;\D_2)$ if and only if $(\G_1;\D_1) \reduce^*
 (\G_2;\D_2)$.
\qed
\end{property}

\begin{figure}[t]
\fbox{%
\parbox{0.99\linewidth}{%
$$
\infer[lts!]
 {(\G;\D, a)  \ar{!a}  (\G;\D)}
 {}
\qquad
\infer[lts?]
 {(\G;\D, a\lolli B)  \ar{?a}  (\G;\D,B)}
 {}
$$

$$
\infer[lts!?]
 {((\G_1;\D_1), (\G_2; \D_2)) \ar{\tau} ((\G_1';\D_1'), (\G_2';\D_2'))}
 {(\G_1;\D_1) \ar{!a} (\G_1';\D_1')
 &(\G_2;\D_2) \ar{?a} (\G_2';\D_2')}
$$

$$
\infer[lts\tensor]
 {(\G; \D, A \tensor B) \ar{\tau} (\G; \D, A, B)}
 {}
\qquad
\infer[lts\unit]
 {(\G; \D,\unit) \ar{\tau} (\G;\D)}
 {}
$$

$$
\infer[lts\with_i]
 {(\G; \D, A_1 \with A_2) \ar{\tau} (\G; \D, A_i)}
 {}
\qquad
\text{(No rule for $\top$)}
$$

$$
\infer[lts!A]
 {(\G; \D,\bang{A}) \ar{\tau} (\G,A;\D)}
 {}
\qquad
\infer[ltsClone]
 {(\G,A;\D) \ar{\tau} (\G,A; \D,A)}
 {}
$$
}}
\caption{Labeled Transition System}
\label{fig:lts-!}
\end{figure}

Based on the labeled transition system in Figure~\ref{fig:lts-!}, we are in a
position to give a coinductive definition of the simulation preorder.
\begin{definition}[Simulation preorder]
\label{def:sim-!}
A relation $\aRel$ between two processes represented as $(\G_1;\D_1)$ and
$(\G_2;\D_2)$ is a \emph{simulation} if $(\G_1;\D_1)\aRel (\G_2;\D_2)$ implies
that
\begin{enumerate}
\itemsep=0ex
\item%
  if $(\G_1;\D_1) \equiv (\G'_1;\cdot)$ then $(\G_2;\D_2)\dar{\tau}
  (\G_2';\cdot)$ and $(\G'_1;\cdot) \aRel (\G'_2;\cdot)$.
\item%
  if $(\G_1;\D_1)\equiv ((\G_1';\D'_1),(\G_1'';\D''_1))$ then $(\G_2;\D_2)
  \dar{\tau} ((\G_2';\D'_2), (\G_2'';\D''_2))$ for some $(\G_2';\D'_2)$ and
  $(\G_2'';\D''_2)$ such that $(\G_1';\D'_1) \aRel (\G_2';\D'_2)$ and
  $(\G_1'';\D''_1) \aRel (\G_2'';\D''_2)$.
\item%
  if $(\G_1;\D_1) \ar{\alpha} (\G_1';\D'_1)$, there exists
  $(\G_2';\D'_2)$ such that $(\G_2;\D_2) \dar{\alpha} (\G_2';\D'_2)$ and
  $(\G_1';\D'_1) \aRel (\G_2';\D'_2)$.
\item%
  if $(\G_1;\D_1) \ar{?a} (\G_1';\D'_1)$, there exists $(\G_2';\D'_2)$
  such that $(\G_2;\:\D_2,a) \dar{\tau} (\G_2';\D'_2)$
and $(\G_1';\D'_1) \aRel (\G_2';\D'_2)$.
\end{enumerate} 
We write $(\G_1;\D_1) \si (\G_2;\D_2)$ if there is some simulation $\aRel$
with $(\G_1;\D_1) \aRel (\G_2;\D_2)$.  We call $\si$ the \emph{simulation
  preorder}.
\qed
\end{definition}
The first two points of the definition ensure that a simulation is
partition-preserving.  The others characterize similarity.  The fourth is a
key bridge to the logical behavior of implication.  It is inspired by the
asynchronous bisimulation proposed in \cite{ACS08}. The intuition is that,
according to the interpretation in Figure~\ref{fig:paf}, formulas are
essentially viewed as processes in an asynchronous version of CCS\@.  The
simulation preorder defined above is reflexive, transitive (i.e., a preorder)
and compositional~\cite[Proposition~4.9, Theorem~4.11,
Proposition~5.5]{deng11tr}:
%???

\begin{property}[$\si$ is a compositional preorder]\label{prot:comp}\hfill
\begin{itemize}
\itemsep=0pt
\item $\si$ is a preorder.
\item%
  If $(\G_1;\D_1) \si (\G_2;\D_2)$, then $((\G_1;\D_1),(\G;\D)) \si
  ((\G_2;\D_2),(\G;\D))$ for any process state $(\G;\D)$.
\qed
\end{itemize}
\end{property}

The soundness and completeness of the contextual preorder with respect to the
simulation preorder is readily established by coinduction~\cite[Theorem~4.12,
4.13, and~5.6]{deng11tr}:

\begin{theorem}\label{thm:cxt,ll}
$(\G_1;\D_1) \cxtleq (\G_2;\D_2)$ if and only if $(\G_1;\D_1) \si
(\G_2;\D_2)$.
\qed
\end{theorem}
\noindent
Here the soundness proof is non-trivial, as it heavily relies on the
compositionality property, shown in the second part in
Property~\ref{prot:comp}.

Relating the simulation preorder and the logical preorder is more involved,
and is where the inductive approach to reasoning about the former meets the
coinductive arguments normally used with the latter.  The rest of this
subsection is indeed devoted to showing the coincidence of $\lleq$ and $\si$.
We first need two technical lemmas whose simple proofs can be found
in~\cite{deng11tr}.

\begin{lemma}[Weakening]
\label{lem:si.ext}
$(\G;\D) \si ((\G,\G');\D)$ for any $\G'$.
\end{lemma}

\begin{lemma}\label{lem:tau.si}
If $(\G_1;\D_1)\dar{\tau} (\G_2;\D_2)$, then $(\G_2;\D_2) \si (\G_1;\D_1)$.
\end{lemma}

We are now in a position to connect simulation with
provability. 
\begin{theorem}\label{thm:prov.si2}
  If $\lseq[]{\G;\D}{A}$, then $(\G;A) \si (\G;\D)$.
\end{theorem}
\begin{proof}
We proceed by rule induction, illustrating three representative rules:
\begin{itemize}
\item (rule clone) %
  Suppose $\lseq[]{\G,B;\D}{A}$ is derived from $\lseq[]{\G,B;\D,B}{A}$. By
  induction, we have
  \begin{equation}\label{eq:clo}
    (\G,B;\:A) \si (\G,B;\: \D,B).
  \end{equation}
  From $(\G,B;\D)$ we have the transition $(\G,B;\D)\ar{\tau} (\G,B;\: \D,B)$.
  By Lemma~\ref{lem:tau.si}, we know that
  \begin{equation}\label{eq:clo2}
    (\G,B;\: \D,B) \si (\G,B;\D).
  \end{equation}
  Combining (\ref{eq:clo}), (\ref{eq:clo2}), and the transitivity of
  similarity, we obtain $(\G,B;\: A) \si (\G,B;\:\D)$.

\item (rule !L) %
  Suppose $\lseq[]{\G;\D,!B}{A}$ is derived from $\lseq[]{\G,B;\D}{A}$. By
  induction, we have
  \begin{equation}\label{eq:cloa}
    (\G,B;\: A) \si (\G,B;\:\D).
  \end{equation}
  From $(\G;\:\D,!B)$ we have the transition $(\G;\: \D,!B)\ar{\tau} (\G,B;\:
  \D)$.  By Lemma~\ref{lem:tau.si}, we know that
  \begin{equation}\label{eq:clob}
    (\G,B;\: \D) \si (\G;\: \D,!B).
  \end{equation}
  By Lemma~\ref{lem:si.ext} we have
  \begin{equation}\label{eq:cloc}
    (\G;\: A) \si (\G,B;\: A).
  \end{equation}
  Combining (\ref{eq:cloa}) -- (\ref{eq:cloc}), and the transitivity of
  similarity, we obtain $(\G;\: A) \si (\G;\: \D,!B)$.

\item (rule !R) %
  Suppose $\lseq[]{\G;\cdot}{!A}$ is derived from $\lseq[]{\G;\cdot}{A}$. By
  induction we have
  \begin{equation}\label{eq:si1}
    (\G; A) \si (\G;\cdot). 
  \end{equation}
  To conclude this case, we construct a relation $\aRel$ as follows
  \[\begin{array}{rcl}
    \aRel & = & \sset{((\G; \D,!A),\ (\G;\D)) \mid \mbox{for any $\D$}}\\
    & & \cup \sset{((\G,A;\D),\ (\G';\D')) \mid \mbox{for any $\D,\D'$
        and $\G'$with } (\G;\D) \si (\G';\D')}\\
    & & \cup \si
  \end{array}\]
  and show by coinduction that $\aRel$ is a simulation, thus
  $\mathalpha{\aRel} \subseteq \mathalpha{\si}$. Since $(\G;\: !A)\aRel
  (\G;\cdot)$, it follows that $(\G;\: !A)\si (\G;\cdot)$. 

  To see that $\aRel$ is a simulation, we pick any pair of states from
  $\aRel$. It suffices to consider the elements from the first two subsets of
  $\aRel$:
  \begin{itemize}
  \item%
    The two states are $(\G;\: \D,!A)$ and $(\G;\D)$ respectively. Let us
    consider any transition from the first state.
    \begin{itemize}
    \item%
      The transition is $(\G;\: \D,!A) \ar{\tau} (\G,A;\:\D)$. This is matched
      up by the trivial transition $(\G;\D)\dar{\tau}(\G;\D)$ because
      $(\G;\D)\si (\G;\D)$ and thus we have $(\G,A;\D) \aRel (\G;\D)$.
    \item%
      The transition is $(\G;\: \D,!A) \ar{\alpha} (\G';\: \D',!A)$ because of
      $(\G;\D)\ar{\alpha} (\G';\D')$. Then the latter transition can match up
      the former because $(\G';\D',!A)\aRel (\G';\D')$.
    \item%
      The transition is $(\G;\D,!A) \ar{?a} (\G;\D',!A)$ because of $(\G;\D)
      \ar{?a} (\G;\D')$. Then we have $(\G;\D,a)\ar{\tau} (\G;\D')$, which is
      a matching transition because we have $(\G;\D',!A) \aRel (\G;\D')$.
    \item%
      If $(\G;\D,!A)$ can be split as $((\G_1;\D_1),(\G_2;\D_2))$, then $!A$
      occurs in either $\D_1$ or $\D_2$. Without loss of generality, we assume
      that $!A$ occurs in $\D_1$. That is, there is some $\D'_1$ such that
      $\D_1\equiv\D'_1,!A$. Then $(\G;\D)\equiv ((\G_1;\D'_1),(\G_2;\D_2))$.
      It is easy to see that $(\G_1;\D_1) \aRel (\G_1;\D'_1)$ and
      $(\G_2;\D_2)\aRel (\G_2;\D_2)$.
    \end{itemize}

  \item%
    The two states are $(\G,A;\D)$ and $(\G';\D')$ respectively with
    \begin{equation}\label{eq:sii}
      (\G;\D) \si (\G';\D').
    \end{equation} 
    Let us consider any transition from the  first state.
    \begin{itemize}
    \item%
      If $\D\equiv \cdot$, then $(\G;\cdot)\si (\G';\D')$. So there exists
      some $\G''$ such that $(\G';\D')\dar{\tau} (\G'';\cdot)$ and
      $(\G;\cdot)\si (\G'';\cdot)$. It follows that $(\G,A;\cdot) \aRel
      (\G'';\cdot)$ as required.
    \item%
      The transition is $(\G,A;\D)\ar{\tau} (\G,A; \D,A)$. We argue that it is
      matched up by the trivial transition $(\G';\D')\dar{\tau}(\G';\D')$.  By
      (\ref{eq:si1}) and the compositionality of $\si$, we obtain
      \begin{equation}\label{eq:si2}
        (\G; \D,A) \si (\G;\D).
      \end{equation}
      By (\ref{eq:sii}) and (\ref{eq:si2}), together with the transitivity of
      similarity, it can be seen that $(\G; \D,A)\si(\G';\D')$, which implies
      $(\G,A; \D,A) \aRel (\G';\D')$.
    \item%
      The transition is $(\G,A;\D)\ar{\alpha} (\G,A;\D'')$ because of
      $(\G;\D)\ar{\alpha} (\G;\D'')$.  By (\ref{eq:sii}) there exist some
      $\G''',\D'''$ such that $(\G';\D')\dar{\alpha} (\G''';\D''')$ and
      $(\G;\D'')\si (\G''';\D''')$. Therefore, $(\G,A;\D'') \aRel
      (\G''';\D''')$ and we have found the matching transition from
      $(\G';\D')$.
    \item%
      The transition is $(\G,A;\D)\ar{?a} (\G,A;\D'')$ because of
      $(\G;\D)\ar{?a} (\G;\D'')$.  By (\ref{eq:sii}) there exist some
      $\G''',\D'''$ such that $(\G';\D',a)\dar{\alpha} (\G''';\D''')$ and
      $(\G;\D'')\si (\G''';\D''')$. Therefore, $(\G,A;\D'') \aRel
      (\G''';\D''')$ and we have found the matching transition from
      $(\G';\D',a)$.
    \item%
      If $(\G,A;\D)$ can be split as $((\G_1;\D_1), (\G_2;\D_2))$, then $A$
      occurs in either $\G_1$ or $\G_2$. Without loss of generality, we assume
      that $A$ occurs in $\G_1$. That is, there is some $\G'_1$ such that
      $\G_1\equiv\G'_1,A$. Then $(\G;\D)\equiv ((\G'_1;\D_1), (\G_2;\D_2))$.
      By (\ref{eq:sii}) we have the transition $(\G';\D')\dar{\tau}
      ((\G_3;\D_3),(\G_4;\D_4))$ for some $(\G_3;\D_3)$ and $(\G_4;\D_4)$ such
      that $(\G'_1;\D_1)\si (\G_3;\D_3)$ and $(\G_2;\D_2) \si (\G_4;\D_4)$. It
      follows that $(\G_1;\D_1)\aRel (\G_3;\D_3)$ and $(\G_2;\D_2) \aRel
      (\G_4;\D_4)$.
      \qed%
    \end{itemize}
  \end{itemize}
\end{itemize}
\renewcommand{\qedsymbol}{}
\end{proof}
\vspace{-3ex}

\begin{corollary}\label{cor:prov.si2}
  If $\lseq[]{\G;\D}{A}$, then $(\cdot; A) \si (\G;\D)$.
\end{corollary}
\begin{proof}
  By Lemma~\ref{lem:si.ext}, Theorem~\ref{thm:prov.si2}, and the transitivity
  of $\si$.
\end{proof}

\begin{proposition}\label{prop:tau2}
  If $(\G_1;\D_1)\dar{\tau} (\G_2;\D_2)$ and $\G_2;\D_2\vdash A$, then
  $\G_1;\D_1 \vdash A$.
\end{proposition}
\begin{proof}
By rule induction.
\end{proof}

Our next goal is to prove the coincidence of logical preorder with simulation.
For that purpose, a series of intermediate results are needed.

\begin{theorem}\label{thm:recomplete2}
  If $(\G_1; A) \si (\G_2;\D)$, then $\G_2;\D \vdash A$.
\end{theorem}
\begin{proof}
  By induction on the structure of $A$. As an example, we consider one case.
  \begin{itemize}
  \item%
    $A\equiv {!}A'$.  By rule lts!A we have the transition $(\G_1; A)\ar{\tau}
    (\G_1,A';\cdot)$. Since $(\G_1,A) \si (\G_2;\D)$ there is some $\G'_2$
    such that $(\G_2;\D)\dar{\tau} (\G'_2;\cdot)$ and
    \begin{equation}\label{eq:tran2}
      (\G_1,A';\cdot) \si (\G'_2;\cdot). 
    \end{equation}
    From $(\G_1,A';\cdot)$ we have the transition $(\G_1,A';\cdot)\ar{\tau}
    (\G_1,A';A')$ by rule ltsClone. By Lemma~\ref{lem:tau.si} we have
    \begin{equation}\label{eq:tran1}
      (\G_1,A'; A') \si (\G_1,A';\cdot)
    \end{equation}
    It follows from (\ref{eq:tran2}), (\ref{eq:tran1}), and the transitivity
    of similarity that
    \begin{equation}\label{eq:tran3}
      (\G_1,A'; A') \si (\G'_2;\cdot). 
    \end{equation}
    Now by induction hypothesis, we obtain $\lseq[]{\G'_2;\cdot}{A'}$ because
    $A'$ has a smaller structure than $A$. By rule !R we infer that
    $\lseq[]{\G'_2;\cdot}{A}$. Using Proposition~\ref{prop:tau2} we conclude
    that $\lseq[]{\G_2;\D}{A}$.  \qed
\end{itemize}
\renewcommand{\qedsymbol}{}
\end{proof}

\begin{theorem}\label{thm:ls2}
$(\G;\D_1) \lleq (\G;\D_2)$ if and only if $(\G;\D_1) \si (\G;\D_2)$.
\end{theorem}
\begin{proof}
  \begin{enumerate}
  \item[($\Rightarrow$)] %
    Suppose $(\G;\D_1)\lleq (\G;\D_2)$. It is trivial to see that
    $\lseq[]{\G;\D_1}{\bigotimes\D_1}$. By the definition of logical preorder, it
    follows that $\lseq[]{\G;\D_2}{\bigotimes\D_1}$. By
    Theorem~\ref{thm:prov.si2}
    we have
    \begin{equation}\label{eq:og1}
      (\G;\bigotimes\D_1) ~\si~ (\G;\D_2).
    \end{equation}
    According to our reduction semantics, we have $(\G;\bigotimes\D_1)
    \dar{\tau} (\G;\D_1)$. By Lemma~\ref{lem:tau.si}, it follows that
    \begin{equation}\label{eq:og2}
      (\G;\D_1) ~\si~ (\G;\bigotimes\D_1).
    \end{equation}
    By combining (\ref{eq:og1}) and (\ref{eq:og2}), we obtain that
    $(\G;\D_1)\si (\G;\D_2)$ because $\si$ is transitive. 

  \item[($\Leftarrow$)] %
    Suppose that $(\G;\D_1)\si (\G;\D_2)$.  For any $\G';\ \D$ and $A$, assume that
    $\lseq[]{(\G',\G;\D,\D_1)}{A}$. By Theorem~\ref{thm:prov.si2} we have
    \begin{equation}\label{eq:og3}
      (\G',\G; A)\si (\G',\G;\D,\D_1). 
    \end{equation}
    Since $(\G;\D_1)\si (\G;\D_2)$ and $\si$ is compositional, we obtain
    \begin{equation}\label{eq:og4}
      (\G',\G; \D,\D_1)~\si~ (\G',\G; \D,\D_2).
    \end{equation}
    By (\ref{eq:og3}), (\ref{eq:og4}) and the transitivity of $\si$, we see that
    $(\G',\G; A) \si (\G',\G;\D,\D_2)$. Then Theorem~\ref{thm:recomplete2} yields
    $\lseq[]{\G',\G;\D,\D_2}{A}$. Therefore, we have shown that
    $(\G;\D_1)\lleq (\G;\D_2)$.
  \qed
  \end{enumerate}
\renewcommand{\qedsymbol}{}
\end{proof}
\vspace{-3ex}

In Theorem~\ref{thm:ls2} we compare two states with exactly the same
unrestricted resource $\G$. The theorem can be relaxed so
that the
two states can have different unrestricted resources. In order to prove
that result, we need two more lemmas. 

\begin{lemma}\label{lem:exp.l}
  $(\G;\D)\lleq (\cdot;\: !\G,\D)$ and $(\cdot;\: !\G,\D) \lleq (\G;\D)$.
\end{lemma}
\begin{proof}
  For any $\G',\D'$, it follows from rule lts!A that $(\G';\D',!\G,\D)
  \dar{\tau} (\G',\G; \D',\D)$. By Proposition~\ref{prop:tau2}, if
  $\G',\G;\D',\D \vdash A$ then $\G';\D',!\G,\D \vdash A$, for any formula
  $A$. In other words, $(\G;\D)\lleq (\cdot;\: !\G,\D)$.

  Suppose $\G';\D',!\G,\D \vdash A$ for any $\G', \D'$ and $A$. By rule
  induction on the derivation of $\G';\D',!\G,\D \vdash A$ it can be shown
  that $\G',\G;\D',\D \vdash A$, thus $(\cdot;\: !\G,\D) \lleq (\G;\D)$.
\end{proof}

\begin{lemma}\label{lem:exp.s}
  $(\G;\D)\si (\cdot; !\G,\D)$ and $(\cdot; !\G,\D) \si (\G;\D)$.
\end{lemma}
\begin{proof}
  Since $(\cdot;\: !\G,\D) \dar{\tau} (\G;\D)$, we apply
  Lemma~\ref{lem:tau.si} and conclude that $(\G,\D)\si (\cdot;\: !\G,\D)$.

  To show that $(\cdot; !\G,\D) \si (\G;\D)$, we let $\aRel$ be the relation
  that relates any state $(\G; !A_1,...,!A_n,\D)$ with the state
  $(\G,A_1,...,A_n;\D )$. The relation $\aRel$ is a simulation. Consider any
  transition from $(\G;!A_1,...,!A_n,\D)$.
  \begin{itemize}
  \item%
    If $(\G; !A_1,...,!A_n,\D) \ar{\alpha} (\G'; !A_1,...,!A_n,\D')$ because
    of $(\G;\D)\ar{\alpha}(\G';\D')$, the transition can be matched up by
    $(\G,A_1,...,A_n;\D ) \ar{\alpha} (\G',A_1,...,A_n;\D')$.
  \item%
    If $(\G;!A_1,...,!A_n;\D) \ar{\tau} (\G,A_1; !A_2,...,!A_n,\D)$, then the
    transition can be matched up by the trivial transition $(\G,A_1,...,A_n;\D
    )\dar{\tau}(\G,A_1,...,A_n;\D )$.
  \item%
    If $(\G;!A_1,...,!A_n,\D)$ performs an input action, it must be given by
    an input action from $\D$. Obviously, this can be mimicked by
    $(\G,A_1,...,A_n;\D )$.
  \item%
    It is easy to see that for any splitting of $(\G;!A_1,...,!A_n,\D)$ there
    is a corresponding splitting of $(\G,A_1,...,A_n;\D )$.
  \end{itemize}
  We have shown that $\aRel$ is a simulation. Therefore,
  $(\G;!A_1,...,!A_n,\D) \si (\G,A_1,...,A_n;\D )$, and as a special case
  $(\cdot;!\G,\D) \si (\G;\D)$.
\end{proof}

We can now establish the correspondence between the logical and simulation preorder.

\begin{theorem}\label{thm:ls3}
  $(\G_1;\D_1) \lleq (\G_2;\D_2)$ if and only if $(\G_1;\D_1) \si (\G_2;\D_2)$.
\end{theorem}
\begin{proof}
  Suppose $(\G_1;\D_1) \lleq (\G_2;\D_2)$. By Lemma~\ref{lem:exp.l} we infer
  that
  \[(\cdot;!\G_1,\D_1)\lleq (\G_1;\D_1) \lleq (\G_2;\D_2) \lleq
  (\cdot;!\G_2,\D_2).\]
  Since $\lleq$ is a preorder, its transitivity gives
$(\cdot;!\G_1,\D_1) \lleq (\cdot;!\G_2,\D_2).$
By Theorem~\ref{thm:ls2}, we have $(\cdot;!\G_1,\D_1) \si
(\cdot;!\G_2,\D_2)$. Then by Lemma~\ref{lem:exp.s} we
infer that
\[(\G_1;\D_1) \si (\cdot; !\G_1,\D_1) \si (\cdot; !\G_2,\D_2) \si
(\G_2;\D_2).\]
By the transitivity of $\si$, we obtain that $(\G_1;\D_1) \si
(\G_2;\D_2)$.

In a similar manner, we can show that $(\G_1;\D_1) \si (\G_2;\D_2)$
implies $(\G_1;\D_1) \lleq (\G_2;\D_2)$.
\end{proof}

Chaining Theorems~\ref{thm:cxt,ll} and \ref{thm:ls3} yields the main result of
the paper, i.e., the equivalence of the logical and contextual preorder.
\begin{corollary}
\label{cor:equiv}
$(\G_1;\D_1) \lleq (\G_2;\D_2)$ if and only if $(\G_1;\D_1) \cxtleq
(\G_2;\D_2)$.
\qed
\end{corollary}

% \section{Ending with a Bang}
% \label{sec:exp}

% \emph{Let linear logic be linear logic!} [RJS --- 29 November 2011]

\section{Conclusions and Future Work}
\label{sec:concl}

Corollary~\ref{cor:equiv} shows that the proof-theoretic notion of logical
preorder coincides with an extensional behavioral relation adapted
from the process-theoretic notion of contextual preorder. The former
is defined exclusively in terms of traditional derivability, and the latter
is defined in terms of a CCS-like
process algebra inspired by the formula-as-process interpretation of a
fragment of linear logic.  In order to establish the connection, a key
ingredient is to introduce a coinductively defined simulation as a stepping
stone. It is interesting to see that coinduction, a central proof technique in
process algebras, is playing an important role in this study of linear logic.
This topic definitely deserves further investigation so that useful ideas
developed in one field can benefit the other, and vice versa.

We have started expanding the results in this paper by examining general
implication (i.e., formulas of the form $A \lolli B$ rather than $a \lolli B$)
and the usual quantifiers.  While special cases are naturally interpreted
into constructs found in the join calculus and the $\pi$-calculus, the
resulting language appears to extend well beyond them.  If successful, this
effort may lead to more expressive process algebras.  We are also interested
in understanding better the interplay of the proof techniques used in the
present work.  This may develop into an approach to employ coinduction
effectively in logical frameworks so as to facilitate formal reasoning and
verification of concurrent systems.

\bibliographystyle{eptcs}
\bibliography{references}
\end{document}